\documentclass[reqno]{amsart}
\usepackage{amsmath, amssymb, amsthm, epsfig}
\usepackage{hyperref, latexsym}
\usepackage{url}
\def\today{\ifcase\month\or
  January\or February\or March\or April\or May\or June\or
  July\or August\or September\or October\or November\or December\fi
  \space\number\day, \number\year}
\newtheorem{Theorem}{Theorem}[section]

\newtheorem{Lemma}[Theorem]{Lemma}

\theoremstyle{definition}

\begin{document}
\title[classical solutions to
Boltzmann with soft potentials]{Distributional and classical solutions to the Cauchy
Boltzmann problem for soft potentials with integrable angular cross section}
\author[R.J. Alonso and I.M. Gamba]{Ricardo J. Alonso and Irene M. Gamba}
\thanks{The authors acknowledge partial  support from NSF grant DMS-0636586 and DMS 0807712. Support from the Institute
from Computational Engineering and Sciences at the University
of Texas at Austin is also gratefully acknowledged.
 }
\date{\today}
\subjclass[2000]{76P05, 35D05}
\keywords{Boltzmann equation for soft potentials; generalized and classical solutions, stability in $L^p$ spaces}
\address{Dept. of Computational \&
Applied Mathematics,
Rice University,
Houston, TX 77005-1892}
\email{Ricardo.J.Alonso@rice.edu}
\address{Department of Mathematics \& ICES, University of Texas at Austin, Austin, TX 78712-1082.}
\email{gamba@math.utexas.edu}
\allowdisplaybreaks
\numberwithin{equation}{section}

%%%%%%%%%%%%%%%%%%%%%%%%%%%%%%%%%%%%%%%%%%%%%%%%%%%%%%%%%%%%%%%%%%%%%%%%%%%%%%%%%%%%%
%%%%%%%%%%%%%%%%%%%%%%%%%%%%%%%%%%%% Abstract %%%%%%%%%%%%%%%%%%%%%%%%%%%%%%%%%%%%%%%
%%%%%%%%%%%%%%%%%%%%%%%%%%%%%%%%%%%%%%%%%%%%%%%%%%%%%%%%%%%%%%%%%%%%%%%%%%%%%%%%%%%%%

\begin{abstract}
This paper focuses on the study of existence and uniqueness of distributional and classical solutions
to the Cauchy Boltzmann problem for the soft potential case assuming $S^{n-1}$ integrability of the
angular part of the collision kernel (Grad cut-off assumption).  For this purpose we revisit the Kaniel--Shinbrot
iteration technique to present an elementary proof of existence and uniqueness results that includes the large data near local Maxwellian regime with possibly infinite initial mass. We study the propagation of regularity using a recent estimate for the positive collision operator given in \cite{alonso-carneiro-gamba}, by E. Carneiro and the authors, that allows us to show such propagation without additional conditions on the collision kernel.  Finally, an $L^{p}$-stability result (with $1\leq p\leq\infty$) is presented assuming the aforementioned condition.
\end{abstract}

\maketitle

%%%%%%%%%%%%%%%%%%%%%%%%%%%%%%%%%%%%%%%%%%%%%%%%%%%%%%%%%%%%%%%%%%%%%%%%%%%%%%%%%%%%%
%%%%%%%%%%%%%%%%%%%%%%%%%%%%% Introduction %%%%%%%%%%%%%%%%%%%%%%%%%%%%%%%%%%%%%%%%%%
%%%%%%%%%%%%%%%%%%%%%%%%%%%%%%%%%%%%%%%%%%%%%%%%%%%%%%%%%%%%%%%%%%%%%%%%%%%%%%%%%%%%%

\section{Introduction}

The purpose of this work is to study  the standard model in the kinetic theory of gases given by the Boltzmann
Transport Equation in the particular case of soft potentials (i.e. collision kernels with singular forms of the relative speed)
under the   $S^{n-1}$ integrability assumption in the
angular part of the collision kernel (Grad cut-off assumption).

More specifically, we study the classical Cauchy Boltzmann problem: Find a function $f(t,x,v)\geq0$ that
solves the equation
\begin{equation}\label{ie1}
\frac{\partial{f}}{\partial{t}}+v\cdot\nabla_x{f}=Q(f,f)\;\;\mbox{in}\;\;(0,+\infty)\times\mathbb{R}^{2n}
\end{equation}
and with initial condition $f(0,x,v)=f_{0}(x,v)$.

Much of the difficulty in solving this problem comes from the complexity of the collision operator $Q$ which
is defined for any two (suitable) functions $f$ and $g$ by the expression
\begin{equation}\label{ie2}
Q(f,g):=\int_{\mathbb{R}^{n}}\int_{S^{n-1}}\left\{f('\!v)g('\!v_{*})-f(v)g(v_{*})\right\}\; B(|u|,\hat{u}\cdot\sigma)\;d\sigma dv_{*},
\end{equation}
where the symbols $'\!v,\,'\!v_{*},\,u$ are defined by
\begin{equation*}
{'\!v}=v-({u}\cdot{\sigma})\;\sigma\ ,\ \ \ {'\!v_{*}}=v_{*}+({u}\cdot{\sigma})\;\sigma\ \ \mbox{and}\ \ \ u=v-v_{*}.
\end{equation*}
In addition, we use  $\hat{x}$ to denote the unitary vector in the direction of $x$,
for $x\in\mathbb{R}^{n}$. The function $B\geq0$ is known as cross section collisional kernel and depends
on the type of interaction
between the particles of the gas. It is usually model as
\begin{equation*}
B(|u|,\hat{u}\cdot\sigma) = {\mathcal G}(|u|) b(\hat{u}\cdot\sigma)\, ,
\end{equation*}
where ${\mathcal G}(|u|)$ is specified below and $b(\hat{u}\cdot\sigma)$ is $S^{n-1}$-integrable with respect to $\sigma$.

For a detailed presentation on the physical meaning of the equation (\ref{ie1})
see the reference \cite{cercignani-illner-pulvirenti}.\\

Solutions to this problem are known to exists in the \textit{renormalized sense} for initial data
with finite mass, energy and entropy.  We refer to \cite{diperna-lions} for the study of renormalized solutions and
to \cite{alexandre-villani}, \cite{hamdache1}, \cite{mischler} and \cite{villani} for further developments in the
theory of very weak solutions.

Existence of distributional solutions (i.e. mild or weak solutions) when the initial data is ``small'' in some sense,
or when is locally ``close'' to the equilibrium (Maxwell distribution) is also well documented.
The theory of distributional solutions for the inhomogeneous Boltzmann case for small initial data
started in the early 80's with the work of Illner and Shinbrot \cite{illner-shinbrot} who adapted the
technique proposed in \cite{kaniel-shinbrot} to produce global in time solutions in the hard potential case.
More recent developments on this theory, including the treatment of the soft potential case, can be found in
\cite{bellomo-toscani}, \cite{hamdache} and \cite{toscani}.\\

There is also some work in the near local equilibrium case (i.e. near a local Maxwellian distribution).
We refer to \cite{palczewski-toscani} for a treatment for the existence and uniqueness of distributional
solutions in the case of hard spheres.  Meanwhile in the references \cite{goudon}, \cite{toscani1} the soft
potential case is discussed including the trend to equilibrium.\\

In the near global equilibrium case (i.e near a global Maxwellian distribution), the literature is focused on the existence and uniqueness of classical solutions in the soft potential case.  The near global equilibrium theory greatly differs from the previous perturbative theories, such as near vacuum and near local Maxwellian, and was started with independent works of Caflisch \cite{caflisch} and Ukai-Asano \cite{ukai-asano} in early 80's.  We refer to \cite{guo1} for recent developments using energy methods to prove existence and uniqueness of classical solutions in the full
soft potential range.

Regarding stability of solutions, Ha has shown the $L^{1}$-stability of solutions,
see \cite{ha0} and \cite{ha1}, in regimes near local Maxwellian distributions as well as near
vacuum.  His results range soft and hard potentials and are valid for the pointwise cutoff
condition $b(\hat{u}\cdot\sigma)\leq K\cos(\hat{u}\cdot\sigma)$.\\

It is worth to mention that the existence of distributional (mild or weak) solutions in the
full soft potential range is an open problem. This problem as well as the propagation of regularity
to produce a global in time estimate on the spatial gradient of
solutions and
the $L^{p}$-stability,  with $1\leq p\leq\infty$,  are the problems treated in this manuscript for initial states
near local Maxwellians or vacuum states, all in the case  without pointwise cut-off or
regularization in the collision kernel.

More specifically, we treat both the distributional and classical theory in a more
unified way for the cases of near vacuum and near local Maxwellian distributions.  We will use techniques that simplify existing proofs and generalize the assumptions on the collision kernel, i.e. avoiding pointwise cut-off or regularization, especially when obtaining classical solutions.  In addition, we reduce the smoothness assumptions on the initial data in order to obtain smooth classical solutions
in both the small data and near local equilibrium cases maintaining minimal regularity on the collision kernel.
As for stability of solutions, we prove under the Grad's cutoff assumption the $L^{p}$-stability ($1\leq p\leq\infty$)
of the solutions found in the aforementioned cases.
The result presented will hold for soft potentials and Maxwell molecules.\\

This paper is organized as follows: In section 2, distributional solutions for the Cauchy Boltzmann problem are constructed under the condition of ``small'' initial data.  A brief discussion of Kaniel-Shinbrot iteration is presented before the main result is proved.  This result of existence and uniqueness applies for both soft and hard potentials.  For the remainder of the paper the proof will be focused only in the soft potential case, thus, in section 3 solutions for the Boltzmann Cauchy problem are constructed for the near equilibrium case.  More specifically, they are build under the assumption that the initial data is ``close'' to the local Maxwellian ($0<\alpha,\;0\leq\beta$)
\begin{equation}\label{lm}
\exp\left(-\alpha\left|x-v\right|^{2}-\beta\left|v\right|^{2}\right).
\end{equation}
The main idea of this discussion is taken from reference \cite{toscani1} and, as in section 2,  uses the Kaniel-Shinbrot iteration approach.  We construct a lower and upper Maxwellian distribution barriers to implement the iteration.  Although our methodology is similar to that of \cite{toscani1}, we have a more direct approach that leads to a relaxation on the conditions imposed to the barriers, in particular, we permit lower and upper barrier with different decay toward infinity.  We note that reference \cite{goudon} also allows different decay for the lower and upper Maxwellian distribution barriers using a functional fixed point argument.  We follow the methodology given in \cite{toscani1} since it is more direct and no additional proof is required to obtain uniqueness of solutions.\\

In section 4, we study the $L^{p}$-propagation of the solution's gradient which permits to construct classical solutions of equation (\ref{ie1}) with natural regularity assumptions on the initial datum, in particular, we assume $\nabla_{x} f_{0}\in L^{p}(\mathbb{R}^{2n})$ for some $1<p<\infty$.  We use a new estimate developed by Carneiro and the authors \cite{alonso-carneiro, alonso-carneiro-gamba} that allows to study the propagation of regularity without pointwise cut-off or regularization in the collision kernel, and to produce a global in time estimate on the spatial gradient of solutions.
In this way we generalize some of the aspects treated in \cite{boudin-desvillettes}.  Furthermore, we address the propagation of regularity in the velocity variable by giving a global in time estimate in the velocity gradient.\\

Finally, two elementary results on stability are shown: stability of solutions in the space of functions uniformly bounded by Maxwellian distributions in the near vacuum case, and a $L^{p}$-stability of solutions ($1\leq p\leq\infty$) in both cases.  The latter is proved under the Grad's cutoff condition, see \textit{(A.2)} below, on the angular kernel using the techniques previously presented in that section.

%%%%%%%%%%%%%%%%%%%%%%%%%%%%%%%%%%%%%%%%%%%%%%%%%%%%%%%%%%%%%%%%%%%%%%%%%%%%%%%%%%%%%%%%%
%%%%%%%%%%%%%%%%%%%%%%%%%%%%%%%  assumption on the model %%%%%%%%%%%%%%%%%%%%%%%%%%%%%%%%
%%%%%%%%%%%%%%%%%%%%%%%%%%%%%%%%%%%%%%%%%%%%%%%%%%%%%%%%%%%%%%%%%%%%%%%%%%%%%%%%%%%%%%%%%
\subsection{Assumption on the model}
Assume that the collision kernel $B(|u|,\hat{u}\cdot\sigma)$ satisfies
\begin{itemize}
\item [\it(A.1)] $B(|u|,\hat{u}\cdot\sigma)=|u|^{-\lambda}\ b(\hat{u}\cdot\sigma)$ with $0\leq\lambda<n-1$.
\item [\it(A.2)] Grad's assumption: $b(\hat{u}\cdot\sigma)\in L^{1}(S^{n-1})$.  We will denote
\begin{displaymath}
\left\|b\right\|_{L^{1}(S^{n-1})}=\int_{S^{n-1}}b(\hat{u}\cdot\sigma)d\sigma.
\end{displaymath}
\end{itemize}
Grad's assumption allows to split the collision operator in a gain and a loss part, namely,
\begin{displaymath}
Q(f,g)=Q_{+}(f,g)-Q_{-}(f,g)
\end{displaymath}
with obvious definitions for each part.  Moreover, the negative part can be expressed as
\begin{displaymath}
Q_{-}(f,g)=f\ R(g),
\end{displaymath}
where
\begin{align}\label{ie3}
R(g)=\int_{\mathbb{R}^{n}}\int_{S^{n-1}}g(v_{*})|u|^{-\lambda}&b(\hat{u}\cdot\sigma)d\sigma dv_{*}\nonumber\\
=\left\|b\right\|_{L^{1}(S^{n-1})}&\int_{\mathbb{R}^{n}}g(v_{*})|u|^{-\lambda}dv_{*}=\left\|b\right\|_{L^{1}(S^{n-1})}\;g\ast|v|^{-\lambda}.
\end{align}
%%%%%%%%%%%%%%%%%%%%%%%%%%%%%%%%%%%%%%%%%%%%%%%%%%%%%%%%%%%%%%%%%%%%%%%%%%%%%%%%%%%%
%%%%%%%%%%%%%%%%%%%%%%%%%%%%%%%%%% Notation and spaces %%%%%%%%%%%%%%%%%%%%%%%%%%%%%
%%%%%%%%%%%%%%%%%%%%%%%%%%%%%%%%%%%%%%%%%%%%%%%%%%%%%%%%%%%%%%%%%%%%%%%%%%%%%%%%%%%%
\subsection{Notation and spaces}
We first introduce some functional spaces where we seek and study solutions to our problem.

Let $M_{\alpha,\beta}(x,v):=\exp\left(-\alpha|x|^{2}-\beta|v|^{2}\right)$ and define the
space of functions bounded by a space-velocity Maxwellian, denoted by $\mathcal{M_{\alpha,\;\beta}}$, as
\begin{equation*}
\mathcal{M_{\alpha,\beta}}=L^{\infty}(\mathbb{R}^{2n},\;M^{-1}_{\alpha,\;\beta})\;,
\end{equation*}
endowed with the norm
\begin{equation*}
\left\|f\right\|_{\alpha,\beta}=\left\|f\;M^{-1}_{\alpha,\beta}\right\|_{L^{\infty}(\mathbb{R}^{2n})}.
\end{equation*}
Let $X$ be a Banach space.  Define $W^{1,1}(0,T;X)$ as those functions $f\in L^{1}(0,T;X)$
such that its time derivative $f_{t}$ exists in the weak sense and belongs to $L^{1}(0,T;X)$.
This space is endowed with the following norm
\begin{equation*}
\left\|f\right\|_{W^{1,1}(0,T;X)}=\left\|f\right\|_{L^{1}(0,T;X)}+\left\|f_{t}\right\|_{L^{1}(0,T;X)}.
\end{equation*}
%%%%%%%%%%%%%%%%%%%%%%%%%%%%%%%%%%%%%%%%%%%%%%%%%%%%%%%%%%%%%%%%%%%%%%%%%%%%%%%%%%%%%%%%%%%
%%%%%%%%%%%%%%%%%%%%%%%%%%%%%%%% Section 1. small initia data %%%%%%%%%%%%%%%%%%%%%%%%%%%%%
%%%%%%%%%%%%%%%%%%%%%%%%%%%%%%%%%%%%%%%%%%%%%%%%%%%%%%%%%%%%%%%%%%%%%%%%%%%%%%%%%%%%%%%%%%%
\section{Distributional solutions for small initial data}
In order to apply the Kaniel and Shinbrot iteration technique
we to introduce the (well known) trajectory operator~$\#$
\begin{equation*}
f^{\#}(t,x,v):=f(t,x+tv,v).
\end{equation*}
This operator gives the evaluation along the trajectories of the transport operator $\partial_{t}+v\cdot\nabla$.
Hence, equation (\ref{ie1}) reduces to
\begin{equation}\label{sde1}
\frac{df^{\#}}{dt}(t)=Q^{\#}(f,f)(t)\;\;\mathrm{with}\;\;f(0)=f_{0}.
\end{equation}
\textbf{Definition.}  A \textit{distributional solution} in $[0,T]$ of problem (\ref{ie1}) is a function $f\in W^{1,1}(0,T;L^{\infty}(\mathbb{R}^{2n}))$ that solves (\ref{sde1}) a.e. in $(0,T]\times\mathbb{R}^{2n}$.\\

Equation (\ref{sde1}) is a good starting point to define the concept of solution because it does not requires differentiability in the \textit{x}-variable for $f$, but equation (\ref{ie1}) does.  Moreover, if $f$ is smooth in the \textit{x}-variable equations (\ref{ie1}) and (\ref{sde1})  the notion of solutions are equivalent in the sense that $f$ is a solution of the former if and only if is a solution of the later.  In other words, equation (\ref{sde1}) is a relaxed version of equation (\ref{ie1}).\\

The concept of distributional (or \textit{mild}) solution is suited to apply a technique introduced at the end of the 70's by Kaniel and Shinbrot \cite{kaniel-shinbrot}.  This technique was first applied for these authors to find local in time mild solution for the Boltzmann equation.  Later, it has been used with success to find global distributional solutions in the context of small initial data for the cases of elastic hard spheres in 3-dimension \cite{illner-shinbrot}, relativistic Boltzmann \cite{glassey}, and recently the inelastic Boltzmann \cite{alonso}.\\

\noindent\textbf{Kaniel and Shinbrot iteration:} In order to present Kaniel and Shinbrot technique we define the sequences $\{l_{n}(t)\}$ and $\{u_{n}(t)\}$ as the mild solutions of the linear problems
\begin{align}
&\frac{dl_{n}^{\#}}{dt}(t)+Q^{\#}_{-}(l_{n},u_{n-1})(t)=Q^{\#}_{+}(l_{n-1},l_{n-1})(t)\ \ \mbox{and}\nonumber\\
&\frac{du_{n}^{\#}}{dt}(t)+Q^{\#}_{-}(u_{n},l_{n-1})(t)=Q^{\#}_{+}(u_{n-1},u_{n-1})(t),\label{sde2}
\end{align}
with $0\leq l_{n}(0)\leq f_{0}\leq u_{n}(0)$.  The construction begins by choosing a pair of functions $(l_{0},u_{0})$ satisfying what Kaniel and Shinbrot called \textit{the beginning condition} in $[0,T]$, i.e. $u^{\#}_{0}\in L^{\infty}(0,T;\mathcal{M_{\alpha,\beta}})$ and
\begin{equation}\label{sde3}
0\leq l^{\#}_{0}(t)\leq l^{\#}_{1}(t)\leq u^{\#}_{1}(t)\leq u^{\#}_{0}(t)\;\;\mbox{a.e. in}\;\; 0\leq t\leq T.
\end{equation}
We summarize the results in \cite{kaniel-shinbrot} with the following theorem
\begin{Theorem}\label{KS}
Let the collision kernel $B$ satisfy the assumptions (A.1)--(A.2) with condition (A.1) relaxed to $-1\leq\lambda<n-1$, and let $\{l_{n}(t)\}$ and $\{u_{n}(t)\}$ be the sequences defined by the mild solutions of the linear problems (\ref{sde2}).  In addition, assume that the beginning condition (\ref{sde3}) is satisfied in $[0,T]$, then
\begin{itemize}
\item [\it(i)] The sequences $\{l_{n}(t)\}$ and $\{u_{n}(t)\}$ are well defined for $n\geq1$.  In addition, $\{l_{n}(t)\}$, $\{u_{n}(t)\}$ are increasing and decreasing sequences respectively, and
\begin{equation*}
l^{\#}_{n}(t)\leq u^{\#}_{n}(t)\;\;\mbox{a.e. in}\;\; 0\leq t\leq T.
\end{equation*}
\item [\it(ii)] If $0\leq l_{n}(0)=f_{0}=u_{n}(0)$ for $n\geq1$, then
\begin{equation*}
\lim_{n\rightarrow\infty}l_{n}(t)=\lim_{n\rightarrow\infty}u_{n}(t)=f(t)\ \ \mbox{a.e. in}\ [0,T].
\end{equation*}
The limit $f(t)\in C(0,T;\mathcal{M}_{\alpha,\beta})$ is the unique distributional solution of the Boltzmann equation in $[0,T]$ and fulfills
\[
0\leq l^{\#}_{0}(t)\leq f^{\#}(t)\leq u^{\#}_{0}(t)\ \mbox{a.e. in}\ [0,T].
\]
\end{itemize}
\end{Theorem}
\noindent The following lemma, which holds for soft and hard potentials, provides the essential
estimate used in the existence of distributional solutions in the near vacuum case.
\begin{Lemma}\label{sdl1}
Assume $-1\leq\lambda<n-1$.  Then, for any $0\leq s\leq t\leq T$ and functions $f^{\#},\ g^{\#}$
that lie in $L^{\infty}(0,T;\mathcal{M}_{\alpha,\beta})$ the following inequality holds
\begin{equation}\label{sde4}\begin{split}
&\int^{t}_{s}\left|Q^{\#}_{+}(f,g)(\tau)\right|d\tau\leq \\
&\qquad\qquad k_{\alpha,\beta}
\exp\left(-\alpha|x|^{2}-\beta|v|^{2}\right)\left\|f^{\#}\right\|_{L^{\infty}(0,T;\mathcal{M}_{\alpha,\beta})}
\left\|g^{\#}\right\|_{L^{\infty}(0,T;\mathcal{M}_{\alpha,\beta})},
\end{split}\end{equation}
where
\[
k_{\alpha,\beta}=\sqrt{\pi}\;\alpha^{-1/2}\;\left\|b\right\|_{L^{1}(S^{n-1})}\left(\frac{|S^{n-1}|}
{n-\lambda-1}+C_{n}\ \beta^{-n/2}\right),
\]
with the constant $C_{n}$ depending only on the dimension.  In other words,
\begin{equation*}
\int^{T}_{0}\left|Q^{\#}_{+}(f,g)(\tau)\right|d\tau\in L^{\infty}(0,T;\mathcal{M}_{\alpha,\beta}).
\end{equation*}
\end{Lemma}
\begin{proof}
An explicit calculation yields the inequality,
\begin{multline}
\left|Q^{\#}_{+}(f,g)(\tau,x,v)\right|\\
\leq e^{-\beta|v|^{2}}\left\|f^{\#}\right\|_{L^{\infty}(0,T;\mathcal{M}_{\alpha,\beta})}
\left\|g^{\#}\right\|_{L^{\infty}(0,T;\mathcal{M}_{\alpha,\beta})}\\
\int_{\mathbb{R}^{n}}e^{-\beta|v_{*}|^{2}}\int_{S^{n-1}}e^{-\alpha|x+\tau(v-'\!v)|^{2}
-\alpha|x+\tau(v-'\!v_{*})|^{2}}b(\hat{u}\cdot\sigma)d\sigma\  |u|^{-\lambda}dv_{*}.\label{sde5}
\end{multline}
Note that
\begin{equation*}
\left|x+\tau(v-'\!v)\right|^{2}+\left|x+\tau(v-'\!v_{*})\right|^{2}=
\left|x\right|^{2}+\left|x+\tau u\right|^{2},
\end{equation*}
and,
\begin{equation*}
\int^{t}_{s}e^{-\alpha\left|x+\tau u\right|^{2}}d\tau\leq\int^{\infty}_{-\infty}e^{-\alpha\left|\tau u\right|^{2}
}d\tau\leq \frac{\sqrt{\pi}}{\alpha^{1/2}}\;|u|^{-1}.
\end{equation*}
Therefore, integrating (\ref{sde5}) in $[s,t]$
\begin{align*}
&\int^{t}_{s}\left|Q^{\#}_{+}(f,g)(\tau,x,v)\right| d\tau\leq\frac{\sqrt{\pi}}{\alpha^{1/2}}
\left\|b\right\|_{L^{1}(S^{n-1})}\exp\left(-\alpha|x|^{2}-\beta|v|^{2}\right)\nonumber\\
&\qquad\qquad\  \left\|f^{\#}\right\|_{L^{\infty}(0,T;\mathcal{M}^{\alpha,\beta})}\left\|g^{\#}
\right\|_{L^{\infty}(0,T;\mathcal{M}^{\alpha,\beta})}\int_{\mathbb{R}^{n}}\exp\left(-\beta|v_{*}|^{2}
\right)\ |u|^{-(\lambda+1)}dv_{*}.
\end{align*}
Finally, the proof is completed by observing that,
\begin{align*}
\int_{\mathbb{R}^{n}}\exp\left(-\beta|v_{*}|^{2}\right)
&\ |u|^{-(\lambda+1)}dv_{*}\nonumber\\
&\leq\int_{\{|v_{*}|<1\}} |u|^{-(\lambda+1)}dv_{*}+\int_{\{|v_{*}|\geq1\}}
\exp\left(-\beta|v_{*}|^{2}\right)dv_{*}\nonumber\\
&\leq\frac{|S^{n-1}|}{n-\lambda-1}+C_{n}\ \beta^{-n/2}.
\end{align*}
\end{proof}

We are now in conditions to prove Theorem \ref{sdexistence} for the global existence of
distributional solutions for soft potentials.  As we previously mentioned, this proof is valid for
both soft and hard potentials as it relies solely on Lemma \ref{sdl1} and Theorem \ref{KS}.
The key step to apply Theorem \ref{KS} is to find suitable functions that satisfy the
\textit{beginning condition} globally.
The most natural (and simplest) choice for the first members is
\begin{equation*}
l^{\#}_{0}=0\;\;\;\;\mathrm{and}\;\;\;\;u^{\#}_{0}=C\exp\left(-\alpha|x|^{2}-\beta|v|^{2}\right).
\end{equation*}
Now compute the following two members
\begin{equation*}
l^{\#}_{1}(t)=f_{0}\exp\left(-\int^{t}_{0}R^{\#}(u_{0})(\tau)d\tau\right)\;\;\;\;
\mathrm{and}\;\;\;\;u^{\#}_{1}(t)=f_{0}+\int^{t}_{0}Q^{\#}_{+}(u_{0},u_{0})(\tau)d\tau.
\end{equation*}
Clearly, $0\leq l^{\#}_{0}\leq l^{\#}_{1}\leq u^{\#}_{1}$.  In addition, using the previous expression
and Lemma \ref{sdl1} we conclude that, for all $t\geq0$,
\begin{equation*}
u^{\#}_{1}(t)\leq \left(\left\|f_{0}\right\|_{\alpha,\beta}+k_{\alpha,\beta}\left\|u^{\#}_{0}\right\|^{2}_{\alpha,\beta}\right)
\exp\left(-\alpha|x|^{2}-\beta|v|^{2}\right).
\end{equation*}
Noting that $\left\|u^{\#}_{0}\right\|_{\alpha,\beta}=C$, therefore it suffices to choose $C$ as
\begin{equation*}
\left\|f_{0}\right\|_{\alpha,\beta}+k_{\alpha,\beta}C^{2}=C
\end{equation*}
in order to satisfy the beginning condition globally.
This is actually  possible for as long as
\begin{equation*}
\left\|f_{0}\right\|_{\alpha,\beta}\leq\frac{1}{4k_{\alpha,\beta}}\ .
\end{equation*}
In particular this previous calculation proves the following theorem.
\begin{Theorem}\label{sdexistence}
Let $B(|u|,\hat{u}\cdot\sigma)$ satisfy the conditions (A.1)--(A.2) with condition (A.1) relaxed to $-1\leq\lambda<n-1$.
 Then, the Cauchy-Boltzmann problem has a unique global distributional solution if
\begin{equation*}
\left\|f_{0}\right\|_{\alpha,\beta}\leq\frac{1}{4k_{\alpha,\beta}}\ ,
\end{equation*}
where the constant $k_{\alpha,\beta}$ is given in Lemma \ref{sdl1}.  Moreover, such  distributional solution satisfies
\begin{equation}\label{e1sd}
\left\|f^{\#}\right\|_{L^{\infty}(0,T;\mathcal{M}_{\alpha,\beta})}\leq C:=\frac{1-\sqrt{1-4k_{\alpha,\beta}\left\|f_{0}\right\|_{\alpha,\beta}}}{2k_{\alpha,\beta}} \, ,
\end{equation}
for any $0\leq T\leq\infty$.
\end{Theorem}

As a consequence of Theorem \ref{sdexistence}, one concludes that the distributional solution $f$ is
controlled by a traveling Maxwellian, and that
\begin{equation*}
\lim_{t\rightarrow\infty}f(t,x,\xi)\rightarrow0\ \mbox{a.e. in}\ \mathbb{R}^{2n}.
\end{equation*}
%%%%%%%%%%%%%%%%%%%%%%%%%%%%%%%%%%%%%%%%%%%%%%%%%%%%%%%%%%%%%%%%%%%%%%%%%%%%%%%%%%%%%%%%%%%%%%%%
%%%%%%%%%%%%%%%%%%%%%%%%%%%%%%%% near Maxwellian existence %%%%%%%%%%%%%%%%%%%%%%%%%%%%%%%%%%%%%
%%%%%%%%%%%%%%%%%%%%%%%%%%%%%%%%%%%%%%%%%%%%%%%%%%%%%%%%%%%%%%%%%%%%%%%%%%%%%%%%%%%%%%%%%%%%%%%%
\section{Distributional solutions near local Maxwellian}
The aim of this section is to use Theorem \ref{KS} to construct solutions for the Cauchy Boltzmann problem
in the soft potential case when the initial data is  locally close to  equilibrium; that is,
the initial data is  near to the local Maxwellian distribution given by (\ref{lm}).  In contrast to the construction for small
data made in the previous section, the negative part of the collision operator will be essential for this
derivation.  The main idea of this construction is taken from \cite{toscani1}. In addition, we refer to
\cite{mischler-perthame} that worked the Maxwellian case with infinite energy and to \cite{goudon} for an
elegant approach with an abstract fixed point argument.\\

Let us introduce some convenient notation and terminology that will help to maintain the proof short and clear.
First, we define the \textit{distance} between two Maxwellian distributions
$M_{i}=C_{i}M_{\alpha_i,\beta_i}$ for $i=1,2$ as
\begin{equation*}
d(M_1,M_2):=|C_{2}-C_{1}|+|\alpha_{2}-\alpha_1|+|\beta_{2}-\beta_{1}|.
\end{equation*}
Second, we say that $f$ is $\epsilon$--\textit{close} to the Maxwellian
distribution $M=C\;M_{\alpha,\beta}$ if there exist Maxwellian distributions $M_{i}$ ($i=1,2$)
such that $d(M_{i},M)<\epsilon\;$ for some small $\epsilon>0$, and
\begin{equation*}
M_{1}\leq f\leq M_{2}.
\end{equation*}
It is also  convenient to define the function
\begin{equation*}
\phi_{\alpha,\beta}(t,x,v):=\left\|b\right\|_{L^{1}(S^{n-1})}\int_{\mathbb{R}^{n}}\exp\left(-\alpha\left|x+u\right|^{2}-\beta
\left|v-u/t\right|^{2}\right)|u|^{-\lambda}du.
\end{equation*}
A simple analysis of this expression shows that  the function
$\phi_{\alpha,\beta}$ is bounded for $-n<\lambda\leq 0$ and that
\begin{equation}\label{pnme1}
\left\|\phi_{\alpha_{2},\beta_{2}}-\phi_{\alpha_{1},\beta_{1}}\right\|_{L^{\infty}}\leq C(\min{\alpha_{i}},\min{\beta_{i}})\;d(M_{1},M_{2}),
\end{equation}
for $0<\alpha_{i}$ and $0<\beta_{i}$.  Similar control holds for $\beta_1=\beta_2=0$ with
constant depending only on $\min{\alpha_i}$.

\begin{Theorem}\label{nmexistence}
Let $B(|u|,\hat{u}\cdot\sigma)$ satisfy the conditions (A.1) and (A.2).  In addition, assume
that $f_{0}(x,v)$ is $\epsilon$--\textit{close} to the local Maxwellian distribution
$C\;M_{\alpha,\beta}(x-v,v)$ ($\alpha>0,\;\beta>0,\;C>0$).  Then, for sufficiently small
$\epsilon$ depending on the initial parameters $C,\alpha,\beta$, the model parameters $\lambda,\left\|b\right\|_{L^{1}(S^{n-1})}$ and
dimension $n$, the Cauchy Boltzmann problem \eqref{ie1} has a unique solution satisfying
\begin{equation}\label{nme0.49}
C_1(t)\;M_{\alpha_1,\beta_1}\left(x-(t+1)v,v\right)\leq f(t,x,v)\leq C_2(t)\;M_{\alpha_2,\beta_2}
\left(x-(t+1)v,v\right),
\end{equation}
for all $t\geq0$, for some continuous functions $0<C_1(t)\leq C\leq C_2(t)<\infty$, and parameters
$0<\alpha_2\leq\alpha\leq\alpha_1$ and $0<\beta_{2}\leq\beta\leq\beta_{1}$.
Moreover, the case $\beta=0$ (infinite mass) is permitted as long as $\beta_1=\beta_2=0$.
\end{Theorem}
\begin{proof}
The idea of the proof is to solve the Cauchy Boltzmann problem for $t\geq 1$, with initial data $f_{0}$ given at $t=1$.
The translation (Galilean) invariance of the Boltzmann equation will imply the result for $t\geq0$.
To this end, we build a lower and upper barriers of the Kaniel and Shinbrot sequences
and prove that they satisfy the \textit{beginning condition}.
Finally, we conclude the result using Theorem \ref{KS}.\\

First note that since $f_{0}(x,v)$ is $\epsilon$--\textit{close} to the local Maxwellian $C\;M_{\alpha,\beta}(x-v,v)$, then, there
 exist two local Maxwellian distributions such that $d(M_{i},M)<\epsilon\;$ ($i=1,2$) and
\begin{equation}\label{efo}
C_{1}\;M_{\alpha_1,\beta_1}(x,v)\leq f_{0}(x+v,v)\leq C_{2}\;M_{\alpha_2,\beta_2}(x,v).
\end{equation}
Since $\epsilon\;$ is expected to be small, we may assume that
\begin{equation}\label{eqp}
\tfrac{1}{2}\;C\leq C_{1}\leq C_{2}\leq 2\;C\ , \ \ \tfrac{1}{2}\;\alpha\leq\alpha_2\leq\alpha_{1}\leq 2\;\alpha\ \ \ \mbox{and}\ \ \ \tfrac{1}{2}\;\beta\leq\beta_2\leq\beta_{1}\leq 2\;\beta.
\end{equation}
Next, define the barriers $l^{\#}_0$ and $u^{\#}_0$ as follows:
\begin{equation}\label{ba}
0\leq l^{\#}_{0}(t)=C_1(t)M_{\alpha_1,\beta_1}\ \ \mbox{and}\ \ u^{\#}_{0}(t)=C_2(t)M_{\alpha_2,\beta_2}.
\end{equation}
The functions $C_1(t)$ and $C_2(t)$ will be chosen later satisfying
\begin{equation}\label{conC}
0\leq C_1(t)\leq C_2(t)\ \ \mbox{with}\ \ C_1(1)=C_1\leq C_2=C_2(1),
\end{equation}
and that the following differential inequalities
\begin{align}
&\frac{dl_{0}^{\#}}{dt}(t)+Q^{\#}_{-}(l_{0},u_{0})(t)\leq Q^{\#}_{+}(l_{0},l_{0})(t)\ \ \ \mbox{and} \label{nme0.5}\\
&\frac{du_{0}^{\#}}{dt}(t)+Q^{\#}_{-}(u_{0},l_{0})(t)\geq Q^{\#}_{+}(u_{0},u_{0})(t),\label{nme1}
\end{align}
are satisfied for all $t\geq1$.
Assume for the moment that such functions $C_{1}(t)$ and $C_{2}(t)$ exist,
then using the definition of $l^{\#}_{1}(t)$ and $u^{\#}_{1}(t)$ in the Kaniel--Shinbrot iteration (\ref{sde2})
\begin{align}
&\frac{dl_{1}^{\#}}{dt}(t)+Q^{\#}_{-}(l_{1},u_{0})(t)=Q^{\#}_{+}(l_{0},l_{0})(t)\label{b1}\\
&\frac{du_{1}^{\#}}{dt}(t)+Q^{\#}_{-}(u_{1},l_{0})(t)=Q^{\#}_{+}(u_{0},u_{0})(t)\label{b2},
\end{align}
with initial condition
\begin{multline}\label{ic1}
0\leq C_1\;M_{\alpha_1,\beta_1}=l^{\#}_{0}(1,x,v)\leq l^{\#}_{1}(1,x,v)=f_{0}(x+v,v)=\\=u^{\#}_{1}(1,x,v)\leq u^{\#}_{0}(1,x,v)=C_2\;M_{\alpha_2,\beta_2},
\end{multline}
one obtains that $l^{\#}_{1}(t)\leq u^{\#}_{1}(t)$.
Since $f_0(x+v,v)\geq C_{1}M_{\alpha_1,\beta_1}$, we can use a classical comparison principle for ODE's between
(\ref{nme0.5}) and (\ref{b1}) to conclude that $l^{\#}_{0}(t)\leq l^{\#}_{1}(t)$.
Similarly,  $f_0(x+v,v)\leq C_{2}M_{\alpha_1,\beta_1}$, thus, comparing (\ref{nme1}) and (\ref{b2}) gives $u^{\#}_1(t)\leq u^{\#}_0(t)$.  Therefore, the \textit{beginning condition} holds for $t\geq1$.  An application of Theorem~\ref{KS} in the interval $[1,\infty)$ and with initial condition $f_{0}(x+v,v)$ produces the desired solution $g(t,x,v)$ satisfying
\begin{equation*}
C_{1}(t)M_{\alpha_1,\beta_1}(x,v)\leq g^{\#}(t,x,v)=g(t,x+tv,v)\leq C_{2}(t)M_{\alpha_1,\beta_1}(x,v)\ \mbox{for}\ t\geq1,
\end{equation*}
with $g^{\#}(1,x,v)=f_{0}(x+v,v)=g(1,x+v,v)$.  This is precisely, up to a shift in time, the control (\ref{nme0.49}) .
The proof concludes by just noticing that $f(t,x,v):=g(t+1,x,v)$ solves the Boltzmann equation in $(0,\infty)$ with initial condition $f_0(x,v)$.\\

It remains to provide two functions $C_{1}(t)$ and $C_2(t)$ fulfilling conditions (\ref{conC}).
First, note that a simple evaluation shows that right hand side in (\ref{nme0.5}) reduces to
\begin{align*}
&Q^{\#}_{+}(l_{0},l_{0})(t)\\
&\qquad =C^{2}_1(t)\;M_{\alpha_1,\beta_1}\left\|b\right\|_{L^{1}(S^{n-1})}\;
\int_{\mathbb{R}^{n}}\exp\left(-\alpha_1\left|x+tu\right|^{2}-\beta_1\left|v-u\right|^{2}\right)
|u|^{-\lambda}du\\
&\qquad =\frac{C^{2}_1(t)}{t^{n-\lambda}}M_{\alpha_1,\beta_1}\;\phi_{1},
\end{align*}
where $\phi_1:=\phi_{\alpha_1,\beta_1}$.  Similarly, the other term in (\ref{nme0.5}) reads
\begin{align*}
&Q^{\#}_{+}(l_{0},u_{0})(t)\\
&\qquad =C_1(t)\;C_2(t)\;M_{\alpha_1,\beta_1}\left\|b\right\|_{L^{1}(S^{n-1})}\;
\int_{\mathbb{R}^{n}}\exp\left(-\alpha_2\left|x+tu\right|^{2}-\beta_2\left|v-u\right|^{2}
\right)|u|^{-\lambda}du\\
&\qquad =\frac{C_1(t)\;C_2(t)}{t^{n-\lambda}}M_{\alpha_1,\beta_1}\;\phi_{2},
\end{align*}
where $\phi_2:=\phi_{\alpha_2,\beta_2}$.  A similar calculation holds for the terms in (\ref{nme1}).

Thus, the differential inequalities (\ref{nme0.5}) and (\ref{nme1}), after the cancellation of the Maxwellian distributions $M_{\alpha_i,\beta_i}$ ($i=1,2$), are recasted as
\begin{align}
&C_1'(t)\leq \frac{C^{2}_1(t)}
{t^{n-\lambda}}\;\phi_{1}-\frac{C_1(t)\;C_2(t)}{t^{n-\lambda}}\;\phi_{2}\label{nme1.8}\\
&C_2'(t)\geq \frac{C^{2}_2(t)}
{t^{n-\lambda}}\;\phi_{2}-\frac{C_1(t)\;C_2(t)}{t^{n-\lambda}}\;\phi_{1}.\label{nme2}
\end{align}
Next, in order to get uniform bounds in space, we want to get rid of the $(x,v)$-dependence in inequalities (\ref{nme1.8}) and (\ref{nme2}).  Indeed, observe that the right hand side of the differential inequality (\ref{nme1.8}) is uniformly controlled in $(x,v)$ from below.
\begin{align}
&\frac{\phi_1\;C_1^{2}(t)-\phi_2\;C_1(t)\;C_2(t)}{t^{n-\lambda}}=\nonumber\\
&\qquad \qquad = \frac{C_1^{2}(t)-C_1(t)\;C_2(t)}{2\;t^{n-\lambda}}\;(\phi_1+\phi_2)+
\frac{C_1^{2}(t)+C_1(t)\;C_2(t)}{2\;t^{n-\lambda}}\;(\phi_1-\phi_2)\nonumber\\
&\qquad \qquad \geq\frac{C_1^{2}(t)-C_1(t)\;C_2(t)}{2\;t^{n-\lambda}}\;\left\|\phi_1+
\phi_2\right\|_{L^{\infty}}-\frac{C_1^{2}(t)+C_1(t)\;C_2(t)}{2\;t^{n-\lambda}}\;
\left\|\phi_1-\phi_2\right\|_{L^{\infty}}.\label{cb}
\end{align}
Similar calculations show that the right hand side of the differential inequality (\ref{nme2}) is controlled by above
\begin{align}
&\frac{\phi_2\;C^{2}_2(t)-\phi_1\;C_1(t)\;C_2(t)}{t^{n-\lambda}}\leq\nonumber\\
&\qquad \qquad \leq \frac{C^{2}_2(t)-C_1(t)\;C_2(t)}{2\;t^{n-\lambda}}\;\left\|\phi_1+
\phi_2\right\|_{L^{\infty}}+\frac{C^{2}_2(t)+C_1(t)\;C_2(t)}{2\;t^{n-\lambda}}\;
\left\|\phi_1-\phi_2\right\|_{L^{\infty}}.\label{ca}
\end{align}

Therefore, we now are in conditions to find  barrier functions  $u^{\#}_0$ and $l^{\#}_{0}$ that satisfy (\ref{nme0.5}) and (\ref{nme1})
with the choice (\ref{ba}).
In fact, a \textit{sufficient condition} for  functions $C_{1}(t)$ and $C_{2}(t)$ to satisfy inequalities (\ref{nme1.8}) and (\ref{nme2}), respectively,  is obtain by choosing them as solutions of the following non linear differential system
\begin{align}
&C_1'(t)=\frac{C_1^{2}(t)-C_1(t)\;C_2(t)}{2\;t^{n-\lambda}}\;\left\|\phi_1+\phi_2\right\|_{L^{\infty}}-
\frac{C_1^{2}(t)+C_1(t)\;C_2(t)}{2\;t^{n-\lambda}}\;\left\|\phi_1-\phi_2\right\|_{L^{\infty}} \label{nme2.9}\\
&C_2'(t)=\frac{C^{2}_2(t)-C_1(t)\;C_2(t)}{2\;t^{n-\lambda}}\;\left\|\phi_1+\phi_2\right\|_{L^{\infty}}
+\frac{C^{2}_2(t)+C_1(t)\;C_2(t)}{2\;t^{n-\lambda}}\;\left\|\phi_1-\phi_2\right\|_{L^{\infty}},\label{nme3}
\end{align}
with initial condition $C_{1}(1)=C_1$ and $C_{2}(1)=C_{2}$.

Indeed, if $C_{1}(t)$ and $C_{2}(t)$ satisfy the differential system (\ref{nme2.9}) and (\ref{nme3}), the uniform control (\ref{cb}) and (\ref{ca}) readily gives that such functions satisfy the differential inequalities (\ref{nme1.8}) and (\ref{nme2}) as well.\\
In \cite{mischler-perthame} the reader can find a proof, for $\lambda=0$ and
$\left\|\phi_1-\phi_2\right\|_{L^{\infty}}=0$ (i.e. $\alpha_1=\alpha_2$ and $\beta_1=\beta_2$),
that such $C_1(t)$ and $C_2(t)$ exist for $t\geq1$ and are uniformly bounded on time provided
that the parameter $\epsilon$ is sufficiently small.  Here, we explicitly present the solution of such differential system by noticing that equations (\ref{nme2.9}), (\ref{nme3}) imply the simple differential relation
\begin{equation*}
\frac{C_1'(t)}{C_1(t)}=-\frac{C_2'(t)}{C_2(t)}\;,
\end{equation*}
which yields the algebraic relation for any $t\geq t_0\geq 1$
\begin{equation*}
\frac{C_1(t)}{C_1(t_{0})}=\frac{C_2(t_{0})}{C_2(t)}\;.
\end{equation*}
Hence, $C_2(t)$ obeys the equation
\begin{align*}
C_2'(t)&+\left(\left\|\phi_1+\phi_2\right\|_{L^{\infty}}-\left\|\phi_1-\phi_2\right\|_{L^{\infty}}\right)
\frac{C_1(1)\;C_2(1)}{2\;t^{n-\lambda}}=\\
&\qquad \qquad= \left(\left\|\phi_1+\phi_2\right\|_{L^{\infty}}+\left\|\phi_1-\phi_2\right\|_{L^{\infty}}\right)
\frac{C_2^{2}(t)}{2\;t^{n-\lambda}}\;,
\end{align*}
which has explicit solutions.

Indeed, let
\begin{equation*}
k^{2}=\frac{\left\|\phi_1+\phi_2\right\|_{L^{\infty}}-\left\|\phi_1-\phi_2\right\|_{L^{\infty}}}
{\left\|\phi_1+\phi_2\right\|_{L^{\infty}}+\left\|\phi_1-\phi_2\right\|_{L^{\infty}}}\;C_1(1)\;C_2(1)\;,
\end{equation*}
then
\begin{equation*}
\frac{C_2(1)+k}{C_2(1)-k}\ \frac{C_2(t)-k}{C_2(t)+k}=\exp\left(k\;\frac{\left\|\phi_1+\phi_2\right\|_{L^{\infty}}+\left\|\phi_1-\phi_2\right\|_{L^{\infty}}}
{n-\lambda-1}\left(1-\frac{1}{t^{n-\lambda-1}}\right)\right)\;.
\end{equation*}
Therefore, $C_2(t)$ will remain uniformly bounded for $t\geq 1$ as long as
\begin{equation}\label{nme4}
\exp\left(k\;\frac{\left\|\phi_1+\phi_2\right\|_{L^{\infty}}+\left\|\phi_1-\phi_2\right\|_{L^{\infty}}}{n-\lambda-1}\right)<
\frac{C_2(1)+k}{C_2(1)-k}\;.
\end{equation}
Using (\ref{pnme1}) and (\ref{eqp}), an elementary calculation shows that
\begin{align*}
\left|C_2(1)-k\right|\leq K_1(C,\alpha,\beta)\;d(M_{1},M_{2})&\leq 2\;K_1(C,\alpha,\beta)\;\epsilon\;,\\  \exp\left(k\;\frac{\left\|\phi_1+\phi_2\right\|_{L^{\infty}}+\left\|\phi_1-\phi_2\right\|_{L^{\infty}}}{n-\lambda-1}\right)&\leq K_{2}(C,\alpha,\beta)\; ,\\
\mbox{and}\ \ \ \ C_2(1)+k&\geq K_{3}(C,\alpha,\beta)\;.
\end{align*}
Hence, inequality (\ref{nme4}) is easily obtained for sufficiently small $\epsilon>0$.  Thus, the existence of continuous $C_{1}(t)$ and $C_{2}(t)$ are assured for $t\geq1$.  Moreover, as simple analysis of the previous expressions shows that $C_{1}(t)$ is non increasing and $C_{2}(t)$ is non decreasing, hence $C_{1}(t)\leq C_{2}(t)$, satisfying condition (\ref{conC}).\\

Finally, note that the norm $\left\|\phi_{\alpha,\beta}\right\|_{L^{\infty}}$ is, in general, controlled by a constant independent of $\beta$, hence by using (\ref{pnme1}) the case $\beta_1=\beta_2=0$ of infinite mass and energy is included in the result.
\end{proof}
\noindent\textbf{Remark:} In fact \textit{item (ii)} of the original version of Theorem \ref{KS} was proved by Kaniel and Shinbrot assuming that $u^{\#}_0(t)$ was summable, however, this is not a major restriction for the result to be true as the authors in \cite{mischler-perthame} proved.

%%%%%%%%%%%%%%%%%%%%%%%%%%%%%%%%%%%%%%%%%%%%%%%%%%%%%%%%%%%%%%%%%%%%%%%%%%%%%%%%%
%%%%%%%%%%%%%%%%%%%% classical solutions %%%%%%%%%%%%%%%%%%%%%%%%%%%%%%%%%%%%%%%%
%%%%%%%%%%%%%%%%%%%%%%%%%%%%%%%%%%%%%%%%%%%%%%%%%%%%%%%%%%%%%%%%%%%%%%%%%%%%%%%%%

\section{Classical solutions}
We prove in this section the existence of classical solutions for the cases presented in the two previous ones.  For this purpose we assume basic regularity in the initial data and then prove that such regularity is propagated through time.  This approach is quite different to the one available by perturbation theory of a \textit{global} Maxwellian  (see \cite{guo1,guo2}), where high regularity in the initial data is crucial for their techniques based in Sobolev embedding theorems.  Essentially, we follow the scheme presented in \cite{boudin-desvillettes} adding some new features as needed.\\

\noindent\textbf{Definition.}  A \textit{classical solution} in $[0,T]$ of problem (\ref{ie1}) is a function such that
\begin{itemize}
\item [\it (i)] $f(t)\in W^{1,1}(0,T;L^{\infty}(\mathbb{R}^{2n}))$
\item [\it (ii)] $\nabla_x f\in L^{1}(0,T;L^{p}(\mathbb{R}^{2n}))$ for some $1\leq p$,
\end{itemize}
which solves the Cauchy Boltzmann problem (\ref{ie1}) a.e. in $[0,T)\times\mathbb{R}^{2n}$.\\

Before presenting the proof, we need the following estimate.
\begin{Theorem}\label{HLS}
Let the collision kernel satisfy assumptions $(A.1)$ and $(A.2)$ in section 1.1. Then,
\begin{itemize}
\item [\it(i)]for $1<p,q,r<\infty$ with $1/p+1/q+\lambda/n=1+1/r$,
\begin{equation}\label{cse1}
\left\|Q_{\pm}(f,g)\right\|_{L^{r}_{v}(\mathbb{R}^{n})}\leq C_{\pm}\left\||u|^{-\lambda}\right\|_{L^{n/\lambda}_{w}(\mathbb{R}^{n})}\;\left\|f\right\|_{L^{p}_{v}(\mathbb{R}^{n})}
\left\|g\right\|_{L^{q}_{v}(\mathbb{R}^{n})}.
\end{equation}
\item [\it(ii)]Moreover, assume that
\begin{equation*}
B(|u|,\hat{u}\cdot\sigma)=\Phi(u)\;b(\hat{u}\cdot\sigma)
\end{equation*}
with $\Phi\in L^{s}(\mathbb{R}^{n})$ radially symmetric and non increasing.  Then, for $1\leq p,q,r\leq\infty$ with $1/p+1/q+1/s=1+1/r$ one can estimate
\begin{equation}\label{cse1.1}
\left\|Q_{\pm}(f,g)\right\|_{L^{r}_{v}(\mathbb{R}^{n})}\leq C_{\pm}\left\|\Phi\right\|_{L^{s}(\mathbb{R}^{n})}\left\|f\right\|_{L^{p}_{v}(\mathbb{R}^{n})}\left\|g\right\|_{L^{q}_{v}(\mathbb{R}^{n})},
\end{equation}
\end{itemize}
where $C_{\pm}=C(n,p,q,r,\left\|b\right\|_{L^{1}(S^{n-1})})$ for $n\geq3$.
\end{Theorem}
\noindent The proof of these two estimates can be found in sections 4 and 5 of \cite{alonso-carneiro-gamba}.  The constants are explicit
\begin{multline*}
C_{+}=C(n,p,q,r) \int^{1}_{-1}\left(\frac{2}{1-s}\right)^{a}\;b(s)(1-s^{2})^{\frac{n-3}{2}}ds\\
\mbox{and}\ \ C_{-}=C(n,p,q,r)\;\left\|b\right\|_{L^{1}(S^{n-1})},
\end{multline*}
for some positive constants $C(n,p,q,r)$ and $a$ described in that reference.  Note that there is a singularity at $s=1$ in the integrand defining $C_+$. However,  this singularity does not create  any problem for the Boltzmann equation since the collision operator can be redefined, by symmetrization,  to this new collision angular kernel
\begin{equation*}
\bar{b}(s):=\left(b(s)+b(-s)\right)\chi\left\{s\leq0\right\}.
\end{equation*}
This symmetrization yields the dependence of the constant $C_{+}$ on $\left\|b\right\|_{L^{1}(S^{n-1})}$ for $n\geq3$.

Finally, observe that  we have used the weak $L^{n/\lambda}$-norm of the soft potential $|u|^{-\lambda}$ in estimate \eqref{cse1}. Also, for any $s>1$ the \textit{weak} $L^{s}$-norm of a function $f$ is defined as
\begin{equation*}
\left\|f\right\|_{L^{s}_{w}(\mathbb{R}^{n})}=\sup_{A}|A|^{-1/s'}\int_{A}|f(x)|dx,
\end{equation*}
where $1/s+1/s'=1$ and $A$ denotes an arbitrary measurable set of finite measure $|A|<\infty$.\\

Next, in order to study the propagation of regularity it is convenient to define for $h>0$ and $\hat{x}\in S^{n-1}$ the following finite difference operator
\begin{equation*}
\left(D_{h,\hat{x}}f\right) (x):=\frac{f(x+h\hat{x})-f(x)}{h}\;,
\end{equation*}
and also define the translation operator
\begin{equation*}
\left(\tau_{h,\hat{x}}f\right)(x):=f(x+h\hat{x}).
\end{equation*}
For notation simplicity we write these operators as $D$ and $\tau$ respectively.

Let $f$ be the distributional solution of the Cauchy Boltzmann problem for either small data or near local Maxwellian case.  Fix $h>0$ and $\hat{x}\in S^{n-1}$ and apply the operator $D$ to both sides of equation (\ref{sde1}) to obtain
\begin{equation}\label{cse1.5}
\frac{d(Df)^{\#}}{dt}(t)=(DQ(f,f))^{\#}(t)=Q^{\#}(Df,f)(t)+Q^{\#}(\tau f,Df)(t).
\end{equation}
Multiplying this equation by
\begin{equation*}
p\ |(Df)^{\#}|^{p-1}\ \mbox{sgn}((Df)^{\#})\ \ \mbox{for}\ \ p>1\ ,
\end{equation*}
and integrating in $\mathbb{R}^{2n}$ we are led to
\begin{equation}\label{cse2}
\frac{d\left\|Df\right\|^{p}_{L^{p}}}{dt}=p\int_{\mathbb{R}^{n}}\int_{\mathbb{R}^{n}}|Df|^{p-1}\mbox{sgn}(Df)\left(Q(Df,f)
+Q(\tau f,Df)\right)dv dx.
\end{equation}
Note first that $\mbox{sgn}(Df)\;Q_{-}(Df,f)\leq0$.  Therefore, using equation (\ref{cse2}) we obtain that for $1<p<\infty$,
\begin{align}\label{cse3}
\frac{d\left\|Df\right\|^{p}_{L^{p}}}{dt}&\leq p\int_{\mathbb{R}^{n}}\left\|Df\right\|^{p-1}_{L^{p}_{v}(\mathbb{R}^{n})}\left(\left\|Q_{+}(Df,f)\right\|_{L^{p}_{v}
(\mathbb{R}^{n})}+
\left\|Q_{+}(\tau f,Df)\right\|_{L^{p}_{v}(\mathbb{R}^{n})}\right.\nonumber\\&\hspace{6.5cm}+\left.\left\|Q_{-}(\tau f,Df)\right\|_{L^{p}_{v}(\mathbb{R}^{n})}\right)dx\nonumber \\
&\leq p\ C\int_{\mathbb{R}^{n}}\left\|Df\right\|^{p}_{L^{p}_{v}(\mathbb{R}^{n})}\left(\left\| f\right\|_{L^{\gamma}_{v}(\mathbb{R}^{n})}+\left\| \tau f\right\|_{L^{\gamma}_{v}(\mathbb{R}^{n})}\right)dx.
\end{align}
In obtaining these inequalities, we have used H\"{o}lder's inequality in the inner integral with exponent $p$ and Theorem~\ref{HLS} with $(p,q,r)=(p,\gamma,p)$ and $(p,q,r)=(\gamma,p,p)$ where $\gamma:=n/(n-\lambda)$.
Moreover, the distributional solution $f(t,x,v)$ is controlled by a traveling Maxwellian, then
\begin{align*}
\left\|f\right\|_{L^{\gamma}_{v}(\mathbb{R}^{n})}\leq \frac{C}{(1+t)^{n/\gamma}}&=\frac{C}{(1+t)^{n-\lambda}}\;,\\
&\mbox{and similarly}\ \ \ \left\|\tau f\right\|_{L^{\gamma}_{v}(\mathbb{R}^{n})}\leq\frac{C}{(1+t)^{n-\lambda}}\;.
\end{align*}
Notice that this estimate is valid for solutions with infinite mass and energy due to estimate (\ref{nme0.5}).
Thus, a Gronwall's argument on (\ref{cse3}) shows that
\begin{equation}\label{cse4}
\left\|Df\right\|_{L^{p}(\mathbb{R}^{2n})}(t)\leq\left\|Df_{0}\right\|_{L^{p}(\mathbb{R}^{2n})}\exp\left(\int^{t}_{0}\tfrac{C}
{(1+s)^{n-\lambda}}
\ ds\right),
\end{equation}
where $C=C(n,p,\lambda,\left\|b\right\|_{L^{1}(S^{n-1})})$.
\begin{Theorem}\label{csexistence}
Fix $0\leq T\leq\infty$ and assume the conditions (A.1)--(A.2) in the collision kernel $B$.  Also, assume that $f_{0}$ satisfies the smallness assumption of Theorem \ref{sdexistence} or is near to a local Maxwellian as in Theorem \ref{nmexistence}.  In addition, assume that $\nabla_x f_{0}\in L^{p}(\mathbb{R}^{2n})$ for some $1<p<\infty$.  Then, there is a unique classical solution $f$ to problem (\ref{ie1}) in the interval $[0,T]$ satisfying the estimates of these theorems, and
\begin{equation}\label{cse5}
\left\|\nabla_x f\right\|_{L^{p}(\mathbb{R}^{2n})}(t)\leq C\left\|\nabla_x f_{0}\right\|_{L^{p}(\mathbb{R}^{2n})}\ \ \ \mbox{for all t }\in [0,T],
\end{equation}
with constant $C=C(n,p,\lambda,\left\|b\right\|_{L^{1}(S^{n-1})})$.
\end{Theorem}
\begin{proof}
Thanks to Theorem \ref{sdexistence} (or Theorem \ref{nmexistence}), we have existence of a unique
distributional solution $f$ to the Cauchy Boltzmann problem.  Estimate (\ref{cse5}) for $\nabla_x f$
follows after sending $h\rightarrow0$ in inequality (\ref{cse4}).
After knowing that the weak gradient exist a.e. we can use the chain rule to obtain that for a.e.
$(t,x,v)\in (0,T)\times\mathbb{R}^{2n}$
\begin{equation*}
\frac{df^{\#}}{dt}(t,x,v)=\left(\frac{\partial{f}}{\partial{t}}+v\cdot\nabla_x{f}\right)(t,x+tv,v)=
Q(f,f)(t,x+tv,v).
\end{equation*}
Thus, $f$ solves equation (\ref{ie1}) a.e. in $(0,T)\times\mathbb{R}^{2n}$.
\end{proof}
\noindent\textbf{Remark:} We can also argue in the following way for the small initial data case:
Impose Maxwellian decay on $\nabla_x f_{0}$ and observe that multiplying equation (\ref{cse1.5}) by
$\mbox{sgn}((Df)^{\#}(t))$, integrating the result in $(0,t)$ and using Lemma \ref{sdl1}
\begin{align*}
&\left\|(Df)^{\#}(t)\right\|_{\alpha,\beta}\leq \\
&\ \leq \left\|Df_{0}\right\|_{\alpha,\beta}+k_{\alpha,\beta}
\left\|(Df)^{\#}\right\|_{L^{\infty}(0,T;\mathcal{M}_{\alpha,\beta})}\;\left(\left\|f^{\#}\right\|_{L^{\infty}
(0,T;\mathcal{M}_{\alpha,\beta})}+\left\|(\tau f)^{\#}\right\|_{L^{\infty}(0,T;\mathcal{M}_{\alpha,\beta})}
\right)\\
&\ \leq \left\|Df_{0}\right\|_{\alpha,\beta}+k_{\alpha,\beta}\left\|(Df)^{\#}\right\|_{L^{\infty}(0,T;
\mathcal{M}_{\alpha,\beta})}\left\|f^{\#}\right\|_{L^{\infty}(0,T;\mathcal{M}_{2\alpha,\beta})}
\left(1+\exp\left(2\alpha\; h^{2}\right)\right).
\end{align*}
Fix $\left\|f_{0}\right\|_{2\alpha,\beta}\leq \frac{3}{16k_{2\alpha,\beta}}$ and use Theorem
\ref{sdexistence} to obtain that the distributional solution fulfils the estimate
\begin{equation*}
\left\|f^{\#}\right\|_{L^{\infty}(0,T;\mathcal{M}_{2\alpha,\beta})}\leq\frac{1}{4k_{2\alpha,\beta}}\;.
\end{equation*}
But $\frac{k_{\alpha,\beta}}{k_{2\alpha,\beta}}=\sqrt{2}$, therefore,
\begin{equation*}
\left\|(Df)^{\#}\right\|_{L^{\infty}(0,T;\mathcal{M}_{\alpha,\beta})}\leq \frac{4\;\left\|Df_{0}
\right\|_{\alpha,\beta}}{4-\sqrt{2}\left(1+\exp\left(2\alpha\; h^{2}\right)\right)}\;.
\end{equation*}
Send $h\rightarrow0$ to conclude that the distributional solution satisfies
\begin{equation*}
\left\|(\nabla_x f)^{\#}\right\|_{L^{\infty}(0,T;\mathcal{M}_{\alpha,\beta})}\leq \frac{2\;\left\|
\nabla_x f_{0}\right\|_{\alpha,\beta}}{2-\sqrt{2}}\;.
\end{equation*}
This procedure yields a classical solution to problem (\ref{ie1}) with Maxwellian decay in its gradient
provided small initial data.  Thus, an iterative use of the argument above can produce classical solutions
as smooth as desired on the condition that the initial datum is smooth with its derivatives having
Maxwellian decay.  Observe that the process is valid for both soft and hard potentials and only uses
the integrability of $b$.

\medskip

\subsection{Velocity regularity}
In this subsection we investigate briefly the propagation of velocity smoothness for the classical
solutions obtained above.  To this end we use the finite difference operator for the $v$-variable
\begin{equation*}
\left(D_{h,\hat{v}}f\right) (v):=\frac{f(v+h\hat{v})-f(v)}{h}\;,
\end{equation*}
for a fix $h>0$ and $\hat{v}\in S^{n-1}$.  Similarly for the translation operator $\tau_{h,\hat{v}}$.
As above, we write such operators as $D$ and $\tau$ for notation simplicity.  Take a classical solution o
f the Cauchy Boltzmann problem $f$ and apply the finite difference operator in equation (\ref{ie1}) to obtain
\begin{multline*}
\frac{d(Df)}{dt}(t)+v\cdot\nabla_x (Df)(t)+\hat{v}\cdot\nabla_x(\tau f)(t)\\=DQ(f,f)(t)=Q(Df,f)(t)+Q(\tau f,Df)(t).
\end{multline*}
This equality follows after using the change of variables $v_{*}\rightarrow v_{*}+h\hat{v}$ in the
collision operator.  Moreover, multiplying this equation by
\begin{equation*}
p\ |(Df)|^{p-1}\ \mbox{sgn}((Df))\ \ \mbox{for}\ \ p>1\ ,
\end{equation*}
and integrating in $\mathbb{R}^{2n}$, we obtain
\begin{align}\label{cse6}
\frac{d\left\|Df\right\|^{p}_{L^{p}}}{dt}(t)&\leq p\int_{\mathbb{R}^{n}}\int_{\mathbb{R}^{n}}|Df|^{p-1}
\mbox{sgn}(Df)\left(Q(Df,f)+Q(\tau f,Df)\right)dv dx\nonumber\\
&\ \hspace{4.2cm} +p\int_{\mathbb{R}^{n}}\int_{\mathbb{R}^{n}}|
\nabla_x(\tau f)||Df|^{p-1}dv dx\nonumber\\
&\leq p\ C\int_{\mathbb{R}^{n}}\left\|Df\right\|^{p}_{L^{p}_{v}(\mathbb{R}^{n})}
\left( \left\|f\right\|_{L^{\gamma}_{v}(\mathbb{R}^{n})}+\left\|\tau f\right\|_{L^{\gamma}_{v}
(\mathbb{R}^{n})}\right)dx \nonumber\\
&\ \hspace{4.2cm} +p\left\|Df\right\|^{p-1}_{L^{p}(\mathbb{R}^{2n})}\left\|\nabla_x(\tau f)\right\|_{L^{p}
(\mathbb{R}^{2n})}\nonumber\\
&\leq \frac{p\ C}{(1+t)^{n-\lambda}}\left\|Df\right\|^{p}_{L^{p}(\mathbb{R}^{2n})}+p\left\|Df\right\|^{p-1}_{L^{p}
(\mathbb{R}^{2n})}\left\|\nabla_x f\right\|_{L^{p}(\mathbb{R}^{2n})}.
\end{align}
In order to quantify the size of the $L^{p}$-norm of $Df$ let $X(t):=\left\|Df\right\|^{p}_{L^{p}
(\mathbb{R}^{2n})}(t)$, then inequality (\ref{cse6}) takes the classical non-linear Ordinary Differential inequality  form
(Bernoulli type)
\begin{equation*}
\frac{dX(t)}{dt}\leq {\tt a}(t)X(t)+{\tt b}(t)X^{\frac{p-1}{p}}(t)\;,
\end{equation*}
with
\begin{equation*}
{\tt a}(t)=\frac{p\ C}{(1+t)^{n-\lambda}}\ \ \ \mbox{and}\ \ \ \ {\tt b}(t)=p\left\|(\nabla_x f)(t)\right\|^{p-1}_{L^{p}(\mathbb{R}^{2n})}.
\end{equation*}
Therefore,
\begin{equation*}
X^{\tfrac{1}{p}}(t)\leq X^{\tfrac{1}{p}}_{0}\exp\left(\tfrac{1}{p}\int^{t}_{0}{\tt a}(s)ds\right)+\frac{1}{p}\int^{t}_{0}
\exp\left(\tfrac{1}{p}\int^{t}_{r}{\tt a}(s)ds\right){\tt b}(r)dr,
\end{equation*}
hence, using  estimate (\ref{cse5})
\begin{equation*}
\left\|Df\right\|_{L^{p}(\mathbb{R}^{2n})}(t)\leq\left(\left\|Df_{0}\right\|_{L^{p}(\mathbb{R}^{2n})}+t\left\|\nabla_x f_{0}\right\|_{L^{p}(\mathbb{R}^{2n})}\right)\exp\left(\int^{t}_{0}\tfrac{C}{1+s^{n-\lambda}}ds\right).
\end{equation*}
Thus, letting $h\rightarrow0$ we conclude the following theorem.
\begin{Theorem}\label{T6}
Let $f$ be a classical solution in $[0,T]$ with $f_{0}$ satisfying the condition of Theorem \ref{sdexistence} or Theorem \ref{nmexistence} and $\nabla_{x} f_{0}\in L^{p}(\mathbb{R}^{2n})$ for some $1<p<\infty$.  In addition assume that $\nabla_{v} f_{0}\in L^{p}(\mathbb{R}^{2n})$.  Then, $f$ satisfies the  estimate
\begin{equation}\label{cse7}
\left\|(\nabla_{v} f)(t)\right\|_{L^{p}(\mathbb{R}^{2n})}\leq C\left(\left\|\nabla_{v} f_{0}\right\|_{L^{p}(\mathbb{R}^{2n})}+t\left\|\nabla_{x} f_{0}\right\|_{L^{p}(\mathbb{R}^{2n})}\right),
\end{equation}
with $C=C(n,p,\lambda,\left\|b\right\|_{L^{1}(S^{n-1})})$ independent of  time.
\end{Theorem}

\medskip

\subsection{$L^{p}$ and $\mathcal{M}_{\alpha,\beta}$ stability}
We finish this last section by presenting a short discussion on the stability of solutions in the $L^{p}$ ($1\leq p\leq\infty$) and $\mathcal{M}_{\alpha,\beta}$ spaces also for $S^{n-1}$ integrability in the
angular part $b(\hat{u}\cdot\sigma)$ of the collision kernel (Grad cut-off assumption).  We refer to \cite{ha0} and \cite{ha1} for a discussion on the $L^{1}$ stability for solutions near vacuum and near Maxwellian for soft and hard potentials under the pointwise condition on the angular part of the collision kernel given by $b(\hat{u}\cdot\sigma)\leq K\cos(\hat{u}\cdot\sigma)$.\\

First, take $f$ and $g$ solutions of the Boltzmann Cauchy problem associated to the initial data $f_{0}$ and $g_{0}$ respectively.  These data fulfill the condition of Theorem \ref{sdexistence} (or Theorem \ref{nmexistence}) so that $f$ and $g$ are controlled by Maxwellian distributions as described in these theorems.  Thus,
\begin{equation*}
\frac{d(f-g)^{\#}}{dt}(t)=Q^{\#}(f,f)(t)-Q^{\#}(g,g)(t)=\frac{1}{2}\left[Q^{\#}(f-g,f+g)-Q^{\#}(f+g,f-g)\right].
\end{equation*}
After multiplying by $|(f-g)^{\#}|^{p-1}\mbox{sgn}((f-g)^{\#})$ with $1<p<\infty$ and following the usual steps we arrive to
\begin{equation*}
\frac{d\left\|f-g\right\|^{p}_{L^{p}}}{dt}(t)\leq C\int_{\mathbb{R}^{n}}\left\|f-g\right\|^{p}_{L^{p}_{v}(\mathbb{R}^{n})}\left\|f+g\right\|_{L^{\gamma}_{v}(\mathbb{R}^{n})}dx.
\end{equation*}
Since $f$ and $g$ are controlled by traveling Maxwellian distributions
\begin{equation}\label{stab1}
\left\|f+g\right\|_{L^{\gamma}_{v}(\mathbb{R}^{n})}\leq \frac{C}{(1+t)^{n-\lambda}}\;.
\end{equation}
Therefore, an application of Gronwall's lemma leads to the following Theorem.

\begin{Theorem}\label{ST}
Let $f$ and $g$ distributional solutions of problem (\ref{ie1}) associated to the initial data $f_{0}$ and $g_{0}$ respectively.  Assume that these data satisfie the condition of Theorem \ref{sdexistence} (or Theorem \ref{nmexistence}).  Then, there exist a constant
$C>0$ independent of time such that
\begin{equation}\label{e1ST}
\left\|f-g\right\|_{L^{p}(\mathbb{R}^{2n})}\leq C\;\left\|f_{0}-g_{0}\right\|_{L^{p}(\mathbb{R}^{2n})}\ \ \ \mbox{with}\ \ \ 1\leq p\leq\infty.
\end{equation}
Moreover, for $f_{0}$ and $g_{0}$ sufficiently small in $\mathcal{M}_{\alpha,\beta}$ it holds
\begin{equation}\label{e2ST}
\left\|(f-g)^{\#}\right\|_{L^{\infty}(0,T;\mathcal{M}_{\alpha,\beta})}\leq C\;\left\|f_{0}-g_{0}\right\|_{L^{\infty}(0,T;\mathcal{M}_{\alpha,\beta})}.
\end{equation}
\end{Theorem}
\begin{proof}
The argument preceding \eqref{stab1} shows the theorem holds for $p\in (1,\infty)$.
Thus it remains to prove  estimate (\ref{e1ST}) just for the cases $p=1$ or $p=\infty$.

To this end we split the soft potential in two radially symmetric and decreasing potentials
\begin{equation*}
|u|^{-\lambda}=\Phi_{1}(u)+\Phi_{2}(u),
\end{equation*}
where $\Phi_{1}(u)\in L^{s}(\mathbb{R}^{n})$ for any $1\leq s< n/\lambda$ and $\Phi_{2}(u)\in L^{\infty}$.  For instance,
\begin{equation*}
\Phi_1(u)=\left(|u|^{-\lambda}-1\right)\chi_{\{|u|\leq1\}}\ \ \mbox{and}\ \ \ \Phi_2(u)=|u|^{-\lambda}-\Phi_1.
\end{equation*}
Then, using (\ref{cse1.1}) in Theorem 4.1 with $(p,q,s,r)=(1,s',s,1)$ and $(p,q,s,r)=(\infty,s',s,\infty)$ we have for any $1\leq s <n/\lambda$
\begin{align*}
\left\|Q(f-g,f+g)\right\|_{L^{p}_{v}(\mathbb{R}^{n})}&\leq  \left\|Q_{\Phi_1}(f-g,f+g)\right\|_{L^{1}_{v}(\mathbb{R}^{n})}+\left\|Q_{\Phi_2}(f-g,f+g)\right\|_{L^{p}_{v}(\mathbb{R}^{n})}\\
&\leq C_s\left\|\Phi_{1}\right\|_{L^{s}_{v}(\mathbb{R}^{n})}\left\|f+g\right\|_{L^{s'}_{v}(\mathbb{R}^{n})}\left\|f-g\right\|_{L^{p}_{v}
(\mathbb{R}^{n})}+\\&\hspace{2.5cm}C\left\|\Phi_{2}\right\|_{L^{\infty}_{v}(\mathbb{R}^{n})}\left\|f+g\right\|_{L^{1}_{v}(\mathbb{R}^{n})}
\left\|f-g\right\|_{L^{p}_{v}(\mathbb{R}^{n})},
\end{align*}
 where $1/s+1/s'=1$.  Using the control by traveling Maxwellians we have
\begin{equation*}
\left\|f+g\right\|_{L^{s'}_{v}(\mathbb{R}^{n})}\leq\frac{C}{(1+t)^{n/s'}}\ \ \ \mbox{and}\ \ \ \ \left\|f+g\right\|_{L^{1}_{v}(\mathbb{R}^{n})}\leq\frac{C}{(1+t)^{n}}.
\end{equation*}
Hence, after previously described computations, we obtain the estimate
\begin{equation*}
\left\|f-g\right\|_{L^{p}}(t)\leq \left\|f_0-g_0\right\|_{L^{p}}+\int^{t}_{0}\frac{C_s}{(1+\tau)^{n/s'}}\left\|f-g\right\|_{L^{p}(\mathbb{R}^{2n})}(\tau)\;d\tau.
\end{equation*}
Choosing any $s\in(\frac{n}{n-1},\frac{n}{\lambda})$ one has that $n/s'>1$. Therefore Gronwall's lemma implies that
\begin{equation*}
\left\|f-g\right\|_{L^{p}(\mathbb{R}^{2n})}\leq C\;\left\|f_0-g_0\right\|_{L^{p}(\mathbb{R}^{2n})}
\end{equation*}
with $C$ independent of time.\\
The proof of estimate (\ref{e2ST}) is a direct consequence of Lemma \ref{sdl1} and estimate (\ref{e1sd}).
\end{proof}

\noindent\textbf{Remark:} Clearly, Theorem \ref{ST} gives uniqueness of solutions $f^{\#}\in L^{\infty}(0,T;\mathcal{M}_{\alpha,\beta})$ (for $\alpha>0$, $\beta>0$). So, in particular, solutions constructed in \cite{goudon} are unique.
%%%%%%%%%%%%%%%%%%%%%%%%%%%%%%%%%%%%%%%%%%%%%%%%%%%%%%%%%%%%%%%%%%%%%%%%%%
%%%%%%%%%%%%%%%%%%% references %%%%%%%%%%%%%%%%%%%%%%%%%%%%%%%%%%%%%%%%%%%
%%%%%%%%%%%%%%%%%%%%%%%%%%%%%%%%%%%%%%%%%%%%%%%%%%%%%%%%%%%%%%%%%%%%%%%%%%


\begin{thebibliography}{100}

\bibitem[1] {alexandre-villani} \textsc{Alexandre, R. and Villani, C.}: On the Boltzmann equation for long-range interactions.
\textit{Comm. Pure Appl. Math.}
\textbf{55}, 30--70 (2002).

\bibitem[2] {alonso} \textsc{Alonso, R.}: Existence of global solutions to the Cauchy problem for the inelastic Boltzmann equation with near-vacuum data. \textit{IUMJ}, \textbf{58}, 999–-1022 (2009).

\bibitem[3] {alonso-carneiro} \textsc{Alonso, R. and Carneiro, E.}: Estimates for the Boltzmann collision operator via radial symmetry and Fourier transform. \textit{To appear in Adv. Math.}

\bibitem[4] {alonso-carneiro-gamba} \textsc{Alonso, R., Carneiro, E. and Gamba, I.}: Convolution inequalities for the Boltzmann collision operator  (arXiv:0902.0507 [math.AP]). Submitted for publication.

\bibitem[5] {bellomo-toscani} \textsc{Bellomo, N. and Toscani, G.}: On the Cauchy problem for the nonlinear Boltzmann equation: global existence, uniqueness and asymptotic behavior.
\textit{J. Math. Phys.} \textbf{26}, 334--338 (1985).

\bibitem[6] {boudin-desvillettes} \textsc{Boudin, L. and Desvillettes, Laurent}: On the singularities of the global small solutions of the full Boltzmann equation.
\textit{Monatsh. Math.}
\textbf{131}, 91--108 (2000).

\bibitem[7] {caflisch} \textsc{Caflisch, R.}: The Boltzmann equation with a soft potential (II).
\textit{Comm. Math. Phys.}
\textbf{74}, 97--109 (1980).

\bibitem[8] {cercignani-illner-pulvirenti} \textsc{Cercignani, C., Illner R. and Pulvirenti, M.}:The mathematical theory of dilute gases.
\textit{Appl. Math. Sci.} Springer-Verlag, Berlin (1994).

\bibitem[9] {diperna-lions} \textsc{Diperna, R. and Lions, P-L.}: On the Cauchy problem for the Boltzmann
equation.
\textit{Ann. Math.}
\textbf{130}, 321--366 (1989).

\bibitem[10] {glassey} \textsc{Glassey, R.}:Global solutions to the Cauchy problem for the relativistic Boltzmann equation with near-vacuum data.
\textit{Comm. Math. Phys.}
\textbf{264}, 705--724 (2006).

\bibitem[11] {goudon} \textsc{Goudon, T.}: Generalized invariant sets for the Boltzmann equation.
\textit{Math. Models Methods Appl. Sci.}
\textbf{7}, 457--476 (1997).

\bibitem[12] {guo1} \textsc{Guo, Y.}:Classical solutions to the Boltzmann equation for molecules with an angular cutoff.
\textit{Arch. Ration. Mech. Anal.}
\textbf{169}, 305--353 (2003).

\bibitem[13] {guo2} \textsc{Guo, Y.}:The Vlasov-Maxwell-Boltzmann system near Maxwellians.
\textit{Invent. Math.}
\textbf{153}, 593--630 (2003).

\bibitem[14] {ha0} \textsc{Ha, S.-Y.}: Nonlinear functionals of the Boltzmann equation
and uniform stability estimates.
\textit{J. Diff. Equat.}
\textbf{215}, 178--205 (2005).

\bibitem[15] {ha1} \textsc{Ha, S.-Y. and Yun S.-B.}: Uniform $L^{1}$-stability estimate of the Boltzmann equation near a local Maxwellian.
\textit{Phys. Nonlinear Phenom.}
\textbf{220}, 79--97 (2006).

\bibitem[16] {hamdache} \textsc{Hamdache, K.}: Existence in the large and asymptotic behavior for the Boltzmann equation.
\textit{Japan. J. Appl. Math.}
\textbf{2}, 1--15 (1985).

\bibitem[17] {hamdache1} \textsc{Hamdache, K.}: Initial boundary value problems for Boltzmann equation.  Global existence of week solutions.
\textit{Arch. Rational Mech. Anal.}
\textbf{119}, 309--353 (1992).

\bibitem[18] {illner-shinbrot} \textsc{Illner, R. and Shinbrot, M.}: The Boltzmann equation, global existence for a rare gas in an infinite vacuum.
\textit{Commun. Math. Phys.}
\textbf{95}, 217--226 (1984).

\bibitem[19] {kaniel-shinbrot} \textsc{Kaniel, S. and Shinbrot, M.}: The Boltzmann equation I. Uniqueness and local existence.
\textit{Commun. Math. Phys.}
\textbf{58}, 65--84 (1978).

\bibitem[20] {mischler} \textsc{Mischler, S.}: On the initial boundary value problem for the Vlasov-Poisson-Boltzmann system.
\textit{Commun. Math. Phys.}
\textbf{210}, 447--466 (2000).

\bibitem[21] {mischler-perthame} \textsc{Mischler, S. and Perthame, B.}: Boltzmann equation with infinite energy: renormalized solutions and distributional solutions for small initial data and initial data close to Maxwellian.
\textit{SIAM J. Math. Anal.}
\textbf{28}, 1015--1027 (1997).

\bibitem[22] {palczewski-toscani} \textsc{Palczewski, A. and Toscani, G.}: Global solution of the Boltzmann equation for rigid spheres and initial data close to a local Maxwellian.
\textit{J. Math. Phys.}
\textbf{30}, 2445--2450 (1989).

\bibitem[23] {toscani} \textsc{Toscani, G.}: On the nonlinear Boltzmann equation in unbounded domains.
\textit{ Arch. Rational Mech. Anal.}
\textbf{95}, 37--49 (1986).

\bibitem[24] {toscani1} \textsc{Toscani, G.}: Global solution of the initial value problem for the Boltzmann equation near a local Maxwellian .
\textit{ Arch. Rational Mech. Anal.}
\textbf{102}, 231--241 (1988).

\bibitem[25] {ukai-asano} \textsc{Ukai, S. and Asano, K.}: On the Cauchy problem of the Boltzmann equation with a soft potential.
\textit{Publ. Res. Inst. Math. Sci.}
\textbf{18}, 477--519(57--99) (1982).

\bibitem[26] {villani} \textsc{Villani, C.}: On a new class of weak solutions to the spatially homogeneous Boltzmann
and Landau equations.
\textit{Arch. Rational. Mech. Anal.}
\textbf{143}, 273--307 (1998).

\end{thebibliography}
\end{document}